\documentclass{llncs}
\usepackage[all]{xy}
\usepackage{qsymbols}
\usepackage{marvosym}
\usepackage{fancybox}
\usepackage{prooftree}
\usepackage{url}

\newcommand{\Player}{\textsf{P}}
\newcommand{\Al}{\mbox{\footnotesize \textsf{A}}}
\newcommand{\Be}{\mbox{\footnotesize \textsf{B}}}

\newcommand{\game}[1]{ \langle #1\rangle}
\newcommand{\strat}[1]{ \llbracket #1\rrbracket}
\newcommand{\prof}[1]{ \langle\!\langle #1\rangle\!\rangle}
\newcommand{\conv}[1]{ #1\downarrow}
\newcommand{\Conv}[1]{ #1\Downarrow}
\newcommand{\GIO}{G_{1,0}}
\newcommand{\GOI}{G_{0,1}}
\newcommand{\SIO}{{\sf S_{1}}}
\newcommand{\SOI}{{\sf S_{0}}}
\newcommand{\AcBes}{{\sf{AcBes}}}
\newcommand{\SAcBes}{{\sf{SAcBes}}}
\newcommand{\BcAes}{{\sf{BcAes}}}
\newcommand{\SBcAes}{{\sf{SBcAes}}}
\newcommand{\sioa}{{s_{1,0,a}}}
\newcommand{\soia}{{s_{0,1,a}}}
\newcommand{\siob}{{s_{1,0,b}}}
\newcommand{\soib}{{s_{0,1,b}}}
\newcommand{\full}[1]{\curlyveedownarrow \hspace*{-8pt}\raisebox{6pt}{\scriptsize \it
    #1}}
\newcommand{\fullp}{\full{p}}
\newcommand{\nat}{\ensuremath{\mathbb{N}}}
\newcommand{\real}{\ensuremath{\mathbb{R}}}
 
\newcommand{\convp}{\ensuremath{\mathop{\vdash\! \raisebox{2pt}{\(\scriptstyle p\)}\! \dashv}}}
\newcommand{\sbis}{\sim_{s}}

\title{A simple case of rationality of escalation}
\author{Pierre Lescanne}
\institute{University of Lyon, \'Ecole normale sup\'erieure de Lyon, CNRS (LIP), \\ 46 all\'ee
d'Italie, 69364 Lyon, France}

\begin{document}
\maketitle

\begin{abstract}
  \hrule

\medskip

Escalation is the fact that in a game (for instance an auction), the agents play
forever.  It is not necessary to consider complex examples to establish its
rationality.  In particular, the $0,1$-game is an extremely simple infinite game in
which escalation arises naturally and rationally.  In some sense, it can be
considered as the paradigm of escalation.  Through an example of economic games, we
show the benefit economics can take of coinduction.

  \medskip

\end{abstract}
\hrule 


\bigskip 

\rightline{\parbox{8cm}{\begin{it} [T]he future of economics is increasingly technical work that
    is founded on the vision that the economy is a complex system.
\end{it}
}}
\medskip
\rightline{David Collander~\cite{colander05}}

\bigskip

Sequential games are the natural framework for decision processes.  In this paper we
study a decision phenomenon called \emph{escalation}.  Finite sequential games (also
known as extensive games) have been introduced by Kuhn~\cite{Kuhn:ExtGamesInfo53} and
subgame perfect equilibria have been introduced by
Selten~\cite{selten65:_spiel_behan_eines_oligop_mit} whereas escalation has been
introduced by Shubik~\cite{Shubik:1971}. Sequential games are games in which each
player plays one after the other (or possibly after herself).  In some specific
infinite games, it has been showed that escalation cannot occur among rational players.  Here
we show on a simple example, the 0,1 game, that this is not the case if one uses
coinduction. In addition the 0,1 game has nice properties which make it an excellent
paradigm of escalation and a good domain of application for coalgebras and
coinduction.

\section{The problem of escalation}
\label{sec:escal}

\rightline{\parbox{8cm}{\begin{it} That ``rational agents'' should \emph{not} engage
      in such [escalation] behavior seems obvious.
    \end{it}}}  %
  \medskip%
  \rightline{Wolfgang Leininger~\cite{leininger89:_escal_and_cooop_in_confl_situat}}

\medskip

Escalation in sequential games is a classic of game theory and  it is admitted that escalation is irrational.  The rationality which we
consider is that given by equilibria.  It has been proved that in finite sequential games,
rationality is obtained by a specific equilibrium called \emph{backward induction} (see Appendix).
More precisely a consequence of Aumann's theorem~\cite{aumann95} says that an agent
takes a rational decision in a finite sequential game if she makes her choice
according to backward induction. In this paper we generalize backward induction into
\emph{subgame perfect equilibria} and we consider naturally that rationality is
reached by subgame perfect equilibria (\textsf{SPE} in short) relying on
Capretta's~\cite{capretta:2007} extension of Aumann's theorem.

\paragraph{What is escalation?} 
In a sequential game, escalation is the possibility that agents take rational
decisions forever without stopping.  This phenomenon has been evidenced by
Shubik~\cite{Shubik:1971} in a game called the \emph{dollar auction}. Without being
very difficult, its analysis is relatively involved, because it requires infinitely
many strategy profiles indexed by
$n`:\nat$~\cite{lescanne09:_from_coind_to_ration_of_escal}.  Moreover in each step
there are two and only two equilibria.  By an observation of the past decisions of
her opponent an agent could get a clue of her strategy and might this way avoid
escalation.  This blindness of the agents is perhaps not completely realistic and was
criticized (see~\cite{DBLP:journals/corr/abs-1112-1185} Section~4.2). In this paper,
we propose an example which is much simpler theoretically and which offers infinitely
many equilibria at each step.  Due to the form of the equilibria, the agent has no
clue on which strategy is taken by her opponent.

\paragraph{Escalation and infinite games.}

Books and
articles~\cite{colman99:_game_theor_and_its_applic,gintis00:_game_theor_evolv,osborne04a,leininger89:_escal_and_cooop_in_confl_situat,oneill86:_inten_escal_and_dollar_auction}
which cover escalation take for granted that escalation is irrational.  Following Shubik, all accept that escalation takes place
and can only take place in an infinite game, but their argument uses a reasoning on
finite games.  Indeed, if one cuts the
infinite game in which escalation is supposed to take place at a finite position, one
gets a finite game, in which the only right and rational decision is to never start
the game, because the only backward induction equilibrium corresponds to not start
playing.  Then the result is extrapolated by the authors to infinite games by making
the size of the game to grow to infinity.  However, it has been known for a
long time at least since
Weierstra\ss{}~\cite{weierstrass72}, that the ``cut and extrapolate'' method is wrong
(see Appendix).  For Weierstra\ss{} this would lead to the conclusion that the
infinite sum of differentiable functions would be differentiable whereas he has
exhibited a famous counterexample.  In the case of infinite structures like infinite
games, the right reasoning is coinduction.  With coinduction we were able to show
that the dollar auction has a rational
escalation~\cite{DBLP:journals/acta/LescanneP12,DBLP:journals/corr/abs-1112-1185}.
Currently, since the tools used generally in economics are pre-coinduction based,
they conclude that bubbles and crises are impossible and everybody's experience has
witnessed the opposite.  Careful analysis done by quantitative economists, like for
instance
Bouchaud~\cite{bouchaud08:_econom,bouchaud03:_theor_finan_risk_deriv_pricin}, have
shown that bursts, which share much similarities with escalation, actually take place
at any time scale.  Escalation is therefore an intrinsic feature of economics.
Consequently, coinduction is the tool that economists who call for a refoundation of
economics \cite{colander05,bouchaud08:_econom} are waiting
for~\cite{winschel12:_privat}.

\paragraph{Structure of the paper}

This paper is structured as follows. In Section~\ref{sec:binary-sequ-games} we
present infinite games, infinite strategy profiles and infinite strategies, then we
describe the 0,1-game in Section~\ref{sec:0-1}. Last, we introduce the concept of
equilibrium (Sections~\ref{sec:subg-pertf-equil} and~\ref{sec:Nash}) and we discuss
escalation (Section~\ref{sec:escalation}). In an appendix, we talk about finite games
and finite strategy profiles.

\section{Two choice sequential games}
\label{sec:binary-sequ-games}

Our aim is not to present a general theory. For this the reader is invited to look
at~\cite{Abramsky:arXiv1210.4537,DBLP:journals/corr/abs-1112-1185,DBLP:journals/acta/LescanneP12}.
But we want to give a taste of infinite sequential games through a very simple
one. This game has two agents and two choices.  To support our claim about the
rationality of escalation, we do not need more features.
In~\cite{DBLP:journals/corr/abs-1112-1185}, we have shown the existence of a big
conceptual gap between finite games and infinite games.  

Assume that the set $\textsf{P}$ of agents is made of  two agents called \textsf{A} and \textsf{B}.  In this
framework, an infinite sequential two choice game has two shapes.  In the first shape, it
is an ending position in which case it boils down to the distribution of the payoffs
to the agents. In other words the game is reduced to a function $f: A"|->" f_{\Al},
B"|->" f_{\Be}$ and we write it $\game{f}$.  In the second shape, it is a generic
game with a set $C$ made of two potential choices:
$d$ or $r$ ($d$ for \emph{down} and $r$ for \emph{right)}.  Since the game is
potentially infinite, it may continue forever.  Thus formally in this most general
configuration a game can be seen as a triple:
\begin{displaymath}
g = \game{p, g_d, g_r}.
\end{displaymath}
where $p$ is an agent and $g_d$ and $g_r$ are themselves games. The subgame $g_d$
corresponds to the down choice, i.e., the choice corresponding to go \emph{down} and
the subgame $g_r$ corresponds to the \emph{right} choice, i.e., the choices
corresponding to go to the right. In other words, we define a functor:
\[\game{~}: \textsf{X} \quad "->" \quad \textsf{Payoff} \quad + \quad
\Player \times \textsf{X} \times \textsf{X}.\]
of which \textsf{Game} is the final coalgebra and where $\Player =
\{\mathsf{A}, \mathsf{B}\}$ and $\textsf{Payoff} =  \real^\Player$.

\subsection{Strategy profiles}
\label{sec:prof}

From a game, one can deduce \emph{strategy profiles} (later we will also say simply
\emph{profiles}), which is obtained by adding a label, at each node, which is a choice
made by the agent.  A choice belong to the set $\{d, r\}$. In other words, a
strategy profile is obtained from a game by adding, at each node, a new label, namely
a choice. Therefore a strategy profile which does not correspond to an ending game is
a quadruple:
\[s = \prof{p,c,s_d,s_r},\] where $p$ is an agent ($\Al$ or $\Be$), $c$ is choice
($d$ or $r$), and, $s_d$ and $s_r$ are two strategy profiles.  The strategy profile
which corresponds to an ending position has no choice, namely it is reduced to a
function $\prof{f} = \prof{{A"|->" f_{\Al}}, {B"|->" f_{\Be}}}$.  From a strategy
profile, one
can build a game by removing the choices:
\begin{eqnarray*}
 \mathsf{game}(\prof{f}) &=& \game{f}\\
\mathsf{game}(\prof{p,c,s_d,s_r}) &=& \game{p,\mathsf{game}(s_d), \mathsf{game}(s_r)}
\end{eqnarray*}
$\mathsf{game}(s)$ is the underlying game of the strategy profile $s$.

Given a strategy
profile $s$, one can associate, by induction, a (partial) payoff function $\widehat{s}$,
as follows:
\begin{displaymath}
\begin{array}{lll}
\textit{~when~} &s = \prof{f}           &\quad \widehat{s} \quad =\quad f\\
\textit{~when~} & s = \prof{p,d,s_d,s_r} &\quad \widehat{s} \quad =  \quad \widehat{s_d} \qquad \\
\textit{~when~} & s = \prof{p,r,s_d,s_r} & \quad\widehat{s} \quad =  \quad\widehat{s_r} \qquad 
\end{array}
\end{displaymath}
$\widehat{s}$ is not defined if its definition runs in an infinite process.  For instance,
if \(s_{\Al, \infty}\) is the strategy profile defined in
Section~\ref{sec:escalation}, $\widehat{s_{\Al, \infty}}$ is not defined.  To ensure
that we consider only strategy profiles where the payoff function is defined we
restrict to strategy profiles that are called \emph{convergent}, written $\conv{s}$
(or sometimes prefixed $\downarrow(s)$) and defined as the least predicate satisfying 
\begin{displaymath}
  s = \prof{f} \quad \vee \quad 
  s = \prof{p,d, s_d, s_r} "=>" \conv{s_d} \ \wedge\ 
  \prof{p,r, s_d, s_r} "=>" \conv{s_r}.
\end{displaymath}
\begin{proposition}
  If  $\conv{s}$, then $\widehat{s}$ is defined.
\end{proposition}
\begin{proof}
  By induction.  If $s=\prof{f}$, then since $\widehat{s} =f$ and $f$ is defined,
  $\widehat{s}$ is defined.

  If $s = \prof{p,c, s_d, s_r}$, there are two cases: $c=d$ or $c=r$. Let us look at
  $c=d$.  If $c=d$, $\widehat{s_d}$ is defined by induction
  and since $\widehat{s}=\widehat{s_d}$, we conclude that $\widehat{s}$ is defined.

  The case $c=r$ is similar.
\end{proof}
As we will consider the payoff function also for subprofiles, we want the payoff
function to be defined on subprofiles as well.  Therefore we define a property
stronger than convergence which we call \emph{strong convergence}.  We say that a
strategy profile $s$ is \emph{strongly convergent} and we write it $\Conv{s}$ if it
is the largest predicate fulfilling the following conditions.
\begin{itemize}
\item $\Conv{\prof{p,c, s_d, s_r}}$ if
\begin{itemize}
  \item $\prof{p,c, s_d, s_r}$ is convergent,
  \item $s_d$ is strongly convergent,
  \item $s_r$ is strongly convergent.
  \end{itemize}
\item $\prof{f}$ is always strongly convergent
\end{itemize}
More formally:
\[\Conv{s} \quad "<=(c)>" \quad s = \prof{f} \ \vee \ (s = \prof{p,c, s_d, s_r} \wedge
\conv{s} \wedge \Conv{s_d} \wedge \Conv{s_r}).\]

There is however a difference between the definitions of $\downarrow$ and
$\Downarrow$. Wherever $\conv{s}$ is defined \emph{by induction}\footnote{Roughly
  speaking a definition by induction %
  works from the basic elements, here the ending games, to constructed elements.}, %
from the ending games to the game, $\Conv{s}$ is defined \emph{by
  coinduction}\footnote{Roughly speaking a definition by coinduction works on
  infinite objects, like an invariant.}. %

Both concepts are based on the fixed-point theorem established by
Tarski~\cite{tarski55}. 
The definition of $\Downarrow$ is typical of infinite games and means that $\Downarrow$ is invariant along the
infinite game.  To make the difference clear between the definitions, we use the
symbol \("<=(i)>"\) for inductive definitions and the symbol \("<=(c)>"\) for
coinductive definitions.  By the way, the definition of the function \textsf{game} is
also coinductive.

We can define  the notion of \emph{subprofile},
written $\precsim$:
\[s' \precsim s \quad "<=(i)>" \quad s' \sbis s \ \vee\ s = \prof{p, c, s_d, s_r} \wedge
(s' \precsim s_d \vee s' \precsim s_r),\] %
where $\sbis$ is the bisimilarity among profiles defined as the largest binary
predicate $s' \sbis s$ such that
\[s' =\prof{f}=s \quad \vee \quad (s' = \prof{p, c, s_d', s_r'}
\wedge  s = \prof{p, c, s_d, s_r} \wedge s_d'\sbis s_d\wedge s_r' \sbis s_r). \]
Notice that since we work with infinite objects, we may have $s\not\sbis s'$ and
${s\precsim s' \precsim s}$.  In other words, an infinite profile can be a strict
subprofile of itself.  This is the case for $\sioa$ and $\siob$ in
Section~\ref{sec:subg-pertf-equil}. 
If a profile is strongly convergent, then the payoffs
associated with all its subprofiles are defined.
\begin{proposition}
  If $\Conv{s_1}$ and if $s_2 \precsim s_1$, then $\widehat{s_2}$ is defined.
\end{proposition}

\subsection{The always modality}

We notice that $\downarrow$ characterizes a profile by a property of the head node,
we would say that this property is local. \(\Downarrow\) is obtained by distributing
the property along the game.  In other words we transform the predicate
\(\downarrow\) and such a predicate transformer is called a \emph{modality}.  Here we
are interested by the modality \emph{always}, also written \(\Box\).

Given a predicate $`F$ on strategy profiles, the predicate $\Box\,P$ is defined
coinductively as follows:

\[\Box\,`F(s) "<=(c)>" `F(s) \wedge s = \prof{p, c, s_d, s_r} "=>" (\Box\,`F(s_d) \wedge
\Box `F(s_r)).\]
The predicate ``is strongly convergent'' is the same as the predicate ``is always
convergent''. 
\begin{proposition}
  \(\Conv{s} \qquad "<=>" \qquad \Box\downarrow(s).\)
\end{proposition}
\subsection{Strategies}
\label{sec:strat}

The coalgebra of \emph{strategies}\footnote{A strategy is not the same as a strategy
  profile, which is obtained as the sum of strategies.} is defined by the
functor \[\strat{~}: X "->" \real^{\textsf{P}} \ + \ (\Player + \textsf{Choice})
\times X \times X\] 
where $\textsf{Choice} =\{d,r\}$.
A strategy of agent~$p$ is a game in which some occurrences of
$p$ are replaced by choices.  A strategy is written $\strat{f}$ or $\strat{x,s_1,
  s_2}$.  By replacing the choice made by agent $p$ by the agent $p$ herself, we can
associate a game with a pair consisting of a strategy and an agent:
\begin{displaymath}
\begin{array}{l@{\quad}c@{\quad}ll@{~}l}
  \mathsf{st2g}(\strat{f}, p) &=& \game{f}\\
  \mathsf{st2g}(\strat{x, st_1, st_2}, p) &=& \mathbf{if~} x \in \Player
  &\mathbf{~then~} &\game{x, \mathsf{st2g}(st_1, p), \mathsf{st2g}(st_2, p)} \\
  &&&\mathbf{~else~} &\game{p, \mathsf{st2g}(st_1, p), \mathsf{st2g}(st_2, p)}.
\end{array}
\end{displaymath}

If a strategy $st$ is really the strategy of agent $p$ it should contain nowhere $p$
and should contain a choice $c$ instead.  In this case we say that $st$ is \emph{full for}
$p$ and we write it $st\fullp$.
\begin{eqnarray*}
  \strat{f} \fullp \\
  \strat{x, st_1, st_2}\fullp &"<=(c)>"& (x`;\textsf{Choice}"=>" x \neq p) \wedge st_1\fullp \wedge st_2\fullp.
\end{eqnarray*}
We can sum strategies to make a strategy profile. But for that we have to assume that
all strategies are full and underlie the same game. In other words, $(st_p)_{p`:\Player}$
is a family of strategies such that:
\begin{itemize}
\item $`A p`:\Player, st_p \fullp$,
\item  there exists a game $g$ such that $`A p`:\Player, \mathsf{st2g}(st_p) = g$.
\end{itemize}
We define 
\begin{math} \displaystyle
  \bigoplus_{p`:\Player} st_p
\end{math}
as follows:
\begin{eqnarray*}
  \bigoplus_{p`:\Player} \strat{f} &=& \prof{f}\\
  \strat{c, st_{p',1}, st_{p',2}} \oplus \bigoplus_{p`:\Player\setminus p'} \strat{p', st_{p,1}, st_{p,2}}  &=& \prof{p',c,   \bigoplus_{p`:\Player} st_{p,1},   \bigoplus_{p`:\Player} st_{p,2}}.
\end{eqnarray*}
We can show that the game underlying all the strategies is the game underlying the
strategy profile which is the sum of the strategies.
\begin{proposition}
  $\mathsf{st2g}(st_{p'},p') = \mathsf{game}(\displaystyle \bigoplus_{p`:\Player} st_p).$
\end{proposition}

\section{Comb games and the 0,1-game}
\label{sec:0-1}

We will restrict to simple games which have the shape of combs,
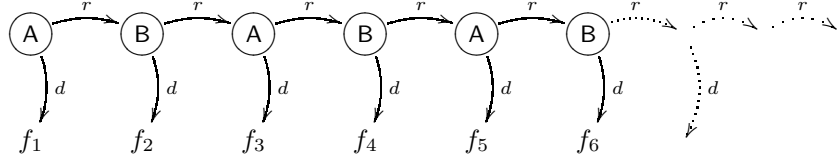
\begin{displaymath}
\newcommand{\pA}[1]{f_{\scriptscriptstyle\mathsf{A},#1}}
\newcommand{\pB}[1]{f_{\scriptscriptstyle\mathsf{B},#1}}
  \xymatrix{
*++[o][F]{\Al} \ar@/^/[r]^r \ar@/^/[d]^{d} &*++[o][F]{\Be} \ar@/^/[r]^r \ar@/^/[d]^{d}
&*++[o][F]{\Al} \ar@/^/[r]^r \ar@/^/[d]^{d} &*++[o][F]{\Be} \ar@/^/[r]^r \ar@/^/[d]^{d} 
&*++[o][F]{\Al} \ar@/^/[r]^r \ar@/^/[d]^{d} &*++[o][F]{\Be} \ar@{.>}@/^/[r]^r \ar@/^/[d]^{d} 
&\ar@{.>}@/^/[r]^r \ar@{.>}@/^/[d]^{d} &\ar@{.>}@/^/[r]^r&\\
f_1&f_2&f_3&f_4&f_5&f_6&&
}
\end{displaymath}
At each step the agents have only two choices, namely to stop or to continue. Let us
call such a game, a \emph{comb game}.

We introduce infinite games by means of equations.  Let us see how this applies to
define the $0,1$-game. First consider two payoff functions:
\begin{eqnarray*}
  f_{0,1} &=& \Al "|->" 0, \Be "|->" 1\\
  f_{1,0} &=& \Al "|->" 1, \Be "|->" 0
\end{eqnarray*}
we define two games 
\begin{eqnarray*}
  \GOI &=& \game{\Al, \game{f_{0,1}}, \GIO}\\
  \GIO &=& \game{\Be, \game{f_{1,0}}, \GOI}
\end{eqnarray*}
This means that we define an infinite sequential game $\GOI$ in which agent $\Al$ is
the first player and which has two subgames: the trivial game $\game{f_{1,0}}$ and
the game $\GIO$ defined in the other equation.  The game $\GOI$ can be pictured as
follows:
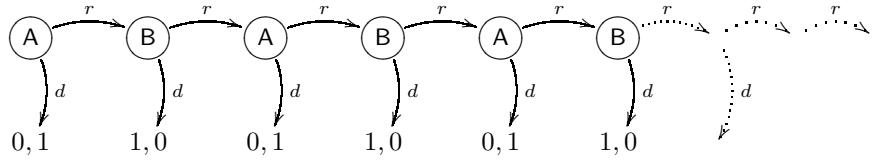
\begin{displaymath}
  \xymatrix{
*++[o][F]{\Al} \ar@/^/[r]^r \ar@/^/[d]^d &*++[o][F]{\Be} \ar@/^/[r]^r \ar@/^/[d]^d
&*++[o][F]{\Al} \ar@/^/[r]^r \ar@/^/[d]^d &*++[o][F]{\Be} \ar@/^/[r]^r \ar@/^/[d]^d 
&*++[o][F]{\Al} \ar@/^/[r]^r \ar@/^/[d]^d &*++[o][F]{\Be} \ar@{.>}@/^/[r]^r \ar@/^/[d]^d 
&\ar@{.>}@/^/[r]^r \ar@{.>}@/^/[d]^d &\ar@{.>}@/^/[r]^r&\\
0,1&1,0&0,1&1,0&0,1&1,0&&
}
\end{displaymath}
or more simply in Figure~\ref{fig:boucle}.a.
\begin{figure}[htb!]
  \centering
  \doublebox{\parbox{\textwidth}{ \begin{displaymath}
      \begin{array}{ccc}
        \xymatrix@C=10pt{
          &\ar@{.>}[r]& *++[o][F]{\Al} \ar@/^1pc/[rr]^r \ar@/^/[d]^d 
          &&*++[o][F]{\Be} \ar@/^1pc/[ll]^r \ar@/^/[d]^d \\
          &&0,1&&1,0
        }
        &
        \xymatrix@C=10pt{
          &\ar@{.>}[r]& *++[o][F]{\Al} \ar@{=>}@/^1pc/[rr]^r \ar@/^/[d]^d 
          &&*++[o][F]{\Be} \ar@/^1pc/[ll]^r \ar@{=>}@/^/[d]^d \\
          &&1,0&&0,1
        } 
        &
        \xymatrix@C=10pt{
          &\ar@{.>}[r]& *++[o][F]{\Al} \ar@/^1pc/[rr]^r \ar@{=>}@/^/[d]^d 
          &&*++[o][F]{\Be} \ar@{=>}@/^1pc/[ll]^r \ar@/^/[d]^d \\
          &&0,1&&1,0
        }\\\\
        a\ (\GOI)  & b \ (\sioa) & c \ (\siob)
      \end{array}
\end{displaymath}}
}
\caption{The $0,1$-game and two equilibria seen compactly}
\label{fig:boucle}
\end{figure}
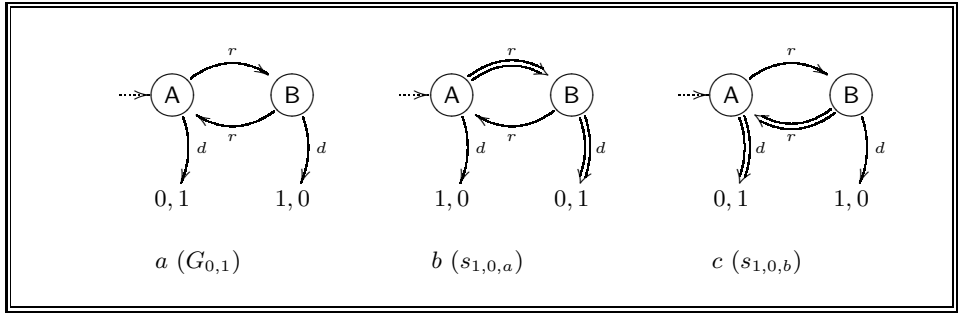
 
From now on, we assume that we consider only strategy profiles whose underlying
game is the 0,1-game. They are characterized by the following predicates
\begin{eqnarray*}
  \SOI(s) &"<=(c)>"& s = \prof{\Al, c, f_{0,1}, s'} \wedge \SIO(s')\\
  \SIO(s) &"<=(c)>"& s = \prof{\Be, c, f_{1,0}, s'} \wedge \SOI(s').
\end{eqnarray*}

Notice that the $0,1$-game we consider is somewhat a zero-sum game, but we are not
interested in this aspect.  Moreover, a very specific instance of a $0,1$ game has
been considered (by Ummels~\cite{DBLP:conf/fossacs/Ummels08} for instance), but these
authors are not interested in the general structure of the game, but in a specific
model on a finite graph, which is not general enough for our taste.  Therefore this
is not a direct generalization of finite sequential games (replacing induction by
coinduction) and this not a framework to study escalation.

\section{Subgame perfect equilibria}
\label{sec:subg-pertf-equil}

Among the strategy profiles, one can select specific ones that are called
\emph{subgame perfect equilibria}.  Subgame perfect equilibria are
specific strategy profiles that fulfill a predicate \textsf{SPE}.  This predicate
relies on another predicate \textsf{PE} which checks a local property.
\[
\begin{array}{ll}
\mathsf{PE}(s) \quad "<=>" \quad \Conv{s}\  &\wedge\ s = \prof{p, d, s_d, s_r} "=>"
\widehat{s_d}(p) \ge \widehat{s_r}(p)\\
& \wedge\ s = \prof{p, r, s_d, s_r} "=>"
\widehat{s_r}(p) \ge \widehat{s_d}(p)
\end{array}
\] %
A strategy profile is a subgame perfect equilibrium if the property \textsf{PE}
holds always:
\[\mathsf{SPE} = \Box \mathsf{PE}.\]
We may now wonder what the subgame perfect equilibria of the 0,1-game are.  We
present two of them in Figure~\ref{fig:boucle}.b and \ref{fig:boucle}.c.  But there
are others.  To present them, let us define a predicate \emph{``\Al{} continues and
  \Be{} eventually stops''}
\begin{eqnarray*}
  \AcBes (s)  "<=(i)>"   s=\prof{p,c,\prof{f}, s'} &"=>"&
  (p =  \Al \wedge f = {f_{0,1}} \wedge c=r \wedge \AcBes(s')) \vee\\ 
  &&(p=\Be \wedge f={f_{1,0}} \wedge (c=d \vee \AcBes(s'))
\end{eqnarray*}
\begin{proposition}\label{prop:AcBes}
$(\SIO(s) \vee \SOI(s))  "=>" \AcBes(s) "=>" \widehat{s} = f_{1,0}$
\end{proposition}
\begin{proof}
  If $s=\prof{p,c,\prof{f}, s'}$, then $\SOI(s') \vee \SIO(s')$.  Therefore if
  $\AcBes(s')$, by induction, $\widehat{s'} = f_{0,1}$. By case:
 \begin{itemize}
 \item If $p=\Al \wedge c=r$, then $\AcBes(s')$ and by definition of $\widehat{s}$, we have $\widehat{s} =
   \widehat{s'}=f_{0,1}$
 \item if $p=\Be \wedge c=d$, the $\widehat{s}=\widehat{\prof{f_{0,1}}} = f_{0,1}$.
 \item if $p=\Be \wedge c=r$, , then $\AcBes(s')$ and by definition of $\widehat{s}$,
   $\widehat{s}=\widehat{s'} = f_{0,1}$.
 \end{itemize}
\end{proof}
Like we generalize \textsf{PE}  to \textsf{SPE} by applying the modality $\Box$, we
generalize $\AcBes$ into $\SAcBes$ by stating:
\[\SAcBes = \Box\AcBes.\]
There are at least two profiles which satisfies $\SAcBes$ namely $\sioa$ and $\siob$
which have been studied in~\cite{DBLP:journals/corr/abs-1112-1185} and pictured in
Figure~\ref{fig:boucle}:
\begin{displaymath}
  \begin{array}{l@{\qquad\qquad}l}
    \begin{array}{lcl}
      \sioa &"<=(c)>"& \prof{\Al, r, f_{0,1}, \siob}\\
      \soia &"<=(c)>"& \prof{\Al, d, f_{0,1}, \soib}
    \end{array}
&
    \begin{array}{lcl}
      \siob &"<=(c)>"& \prof{\Be, d, f_{1,0}, \sioa}\\
      \soib &"<=(c)>"& \prof{\Be, r, f_{1,0}, \soia}
    \end{array}
  \end{array}
\end{displaymath}
\begin{proposition}\label{prop:SAcBes-conv}
  \(\SAcBes(s) "=>" \Conv{s}.\)
\end{proposition}
We may state the following proposition.  
\begin{proposition}
  \(`A s, (\SOI(s) \vee \SIO(s)) "=>" (\SAcBes(s) "=>"  \mathsf{SPE}(s)).\)
\end{proposition}
\begin{proof}
Since \textsf{SPE} is a coinductively defined predicate, the proof is by coinduction.

  Given an $s$, we have to prove \(`A s, \Box \AcBes(s) \wedge (\SOI(s) \vee \SIO(s))
  "=>" \Box\mathsf{PE}(s).\)

  For that we assume $\Box \AcBes(s) \wedge (\SOI(s) \vee \SIO(s))$ and in addition
  (coinduction principle) \(\Box\mathsf{PE}(s')\) for all strict subprofiles $s'$ of
  $s$ and we prove $\mathsf{PE}(s)$. In other words, $\Conv{s} \wedge \prof{p, d, s_d,
    s_r} "=>" \widehat{s_d}(p) \ge \widehat{s_r}(p)\ \wedge\ \prof{p, r, s_d, s_r}
  "=>" \widehat{s_r}(p) \ge \widehat{s_d}(p). $

By Proposition~\ref{prop:SAcBes-conv}, we have $\Conv{s}$.

By Proposition~\ref{prop:AcBes}, we know that for every subprofile $s'$ of a profile
$s$ that satisfies $\SIO(s) \vee \SOI(s)$ we have $\widehat{s'} = f_{1,0}$ except
when $s'=\prof{f_{0,1}}$.  Let us prove $\prof{p, d, s_d, s_r} "=>" \widehat{s_d}(p)
\ge \widehat{s_r}(p)\ \wedge\ \prof{p, r, s_d, s_r} "=>" \widehat{s_r}(p) \ge
\widehat{s_d}(p). $ Let us proceed by case:
  \begin{itemize}
  \item $s=\prof{\Al, r, \prof{f_{0,1}}, s'}$. Then $\SOI(s)$ and $\SIO(s')$. Since
    $\Box \AcBes(s)$, we have $\AcBes(s')$, therefore $\widehat{s'} = f_{1,0}$ hence
    $\widehat{s'}(\Al) = 1$ and $f_{0,1}(\Al) = 0$, henceforth $\widehat{s'}(\Al) \ge
    f_{0,1}(\Al).$
 \item $s=\prof{\Be, r, \prof{f_{1,0}}, s'}$. Then $\SIO(s)$ and $\SOI(s')$. Since
    $\Box \AcBes(s)$, we have $\AcBes(s')$, therefore $\widehat{s'} = f_{1,0}$ hence
    $\widehat{s'}(\Be) = 0$ and $f_{1,0}(\Be) = 0$, henceforth $\widehat{s'}(\Be) \ge
    f_{1,0}(\Be).$
  \end{itemize}
\end{proof}
Symmetrically we can define a predicate $\BcAes$ for ``$\Be$ continues and $\Al$
eventually stops'' and a predicate $\SBcAes$ which is $\SBcAes = \Box\ \BcAes$ which
means that $\Be$ continues always and $\Al$ stops infinitely often. With the same
argument as for $\SAcBes$ one can conclude that 
\[`A s, (\SOI(s) \vee \SIO(s)) "=>" \SBcAes(s) "=>" \mathsf{SPE}(s).\]%
We claim that $\SAcBes \vee \SBcAes$ fully characterizes \textsf{SPE} of
0,1-games, in other words.
\begin{conjecture}
  \(`A s, (\SOI(s) \vee \SIO(s)) "=>" (\SAcBes(s) \vee \SBcAes "<=>"
  \mathsf{SPE}(s)).\)
\end{conjecture}

\section{Nash equilibria}
\label{sec:Nash}
Before talking about escalation, let us see the connection between subgame perfect
equilibrium and Nash equilibrium in a sequential game.  In \cite{osborne04a}, the
definition of a Nash equilibrium is as follows: \textit{A Nash equilibrium is
  a``pattern[s] of behavior with the property that if every player knows every other
  player's behavior she has not reason to change her own behavior''} in other words,
\textit{``a Nash equilibrium [is] a strategy profile from which no player wishes to
  deviate, given the other player's strategies.'' }.  The concept of deviation of
agent $p$ is expressed by a binary relation we call
\emph{convertibility}\footnote{This should be called perhaps \emph{feasibility}
  following~\cite{rubinstein06:microec} and
  \cite{lescanne09:_feasib_desir_games_normal_form}} and we write \convp.  It is
defined inductively as follows:

    \[ \prooftree s\sbis s'
    \justifies s \convp s'
    \endprooftree
    \]

 \[
    \prooftree s_1 \convp s_1'\qquad s_2 \convp s_2' %
    \justifies \prof{p, c, s_1, s_2} \ \convp \ \prof{p, c', s_1', s_2'}
    \endprooftree
    \]

    \[
    \prooftree 
  s_1 \convp s_1'\qquad s_2 \convp s_2' %
  \justifies \prof{p', c, s_1, s_2 } \ \convp \ \prof{p', c, s_1', s_2'}
  \endprooftree
  \]

We define the predicate \textsf{Nash} as follows:
\[\mathsf{Nash}(s) "<=>" `A p,`A s', s \convp s' "=>" \widehat{s}(p) \ge
\widehat{s'}(p').\]

The concept of Nash equilibrium is more general than that of subgame perfect
equilibrium and we have the following result:
\begin{proposition}
  $\mathsf{SPE}(s) "=>" \mathsf{Nash}(s)$.
\end{proposition}
The result has been proven in COQ and we refer to the script
(see\cite{DBLP:journals/acta/LescanneP12}):

\centerline{\url{http://perso.ens-lyon.fr/pierre.lescanne/COQ/EscRatAI/}}

\centerline{\url{http://perso.ens-lyon.fr/pierre.lescanne/COQ/EscRatAI/SCRIPTS/}}

Notice that we defined the
convertibility inductively, but a coinductive definition is possible.  But this would
give a  more restrictive definition of Nash equilibrium.

\section{Escalation}
\label{sec:escalation}

Escalation in a game with a set $\textsf{P}$ of agents occurs when there is a tuple of
strategies $(st_p)_{p`:\Player}$ such that its sum is not convergent, in other words,
$\neg~ \conv{\displaystyle (\bigoplus_{p`:\Player} st_p)}$. Said differently, it is
possible that the agents have all a private strategy which combined with those of the
others makes a strategy profile which is not convergent, which means that the
strategy profile goes to infinity when following the choices.  Notice the two uses of
a strategy profile: first, as a subgame perfect equilibrium, second as a combination
of the strategies of the agents.

Consider the strategy:
\begin{eqnarray*}
  st_{\Al,\infty} &=& \strat{r, \strat{f_{0,1}}, st_{\Al,\infty}'}\\
  st_{\Al,\infty}' &=& \strat{\Be, \strat{f_{1,0}}, st_{\Al,\infty}}
\end{eqnarray*}
and its twin
\begin{eqnarray*}
  st_{\Be,\infty} &=& \strat{\Al, \strat{f_{0,1}}, st_{\Be,\infty}'}\\
  st_{\Be,\infty}' &=& \strat{r, \strat{f_{1,0}}, st_{\Be,\infty}}.
\end{eqnarray*}
Moreover, consider the strategy profile:
\begin{eqnarray*}
  s_{\Al, \infty} & = & \prof{\Al, r, \prof{f_{0,1}}, s_{\Be, \infty}}\\
  s_{\Be, \infty} & = & \prof{\Be, r, \prof{f_{1,0}}, s_{\Al, \infty}}.
\end{eqnarray*}
\begin{proposition}~
  \begin{itemize}
  \item $st_{\Al,\infty}\full{\textsf{\scriptsize \sf A}}$,
 \item $st_{\Be,\infty}\full{\textsf{\scriptsize \sf B}}$,
  \item $\mathsf{st2g}(st_{\Al,\infty}, \Al) = \mathsf{st2g}(st_{\Be,\infty}, \Be) = \GOI,$
  \item $\mathsf{game}(s_{\Al, \infty}) = \GOI$,
  \item $st_{\Al,\infty} \oplus st_{\Be,\infty} = s_{\Al, \infty}$,
  \item $\neg\ \conv{s_{\Al, \infty}}$.
  \end{itemize}
\end{proposition}
\begin{proof}
  By coinduction.
\end{proof}

$ st_{\Al,\infty}$ and $ st_{\Be,\infty}$ are both rational since they are built
using choices, namely~$r$, dictated by subgame perfect equilibria\footnote{Our choice
  of rationality is this of a subgame perfect equilibrium, as it generalizes backward
  induction, which is usually accepted as the criterion of rationality for finite
  game.}  which start with $r$.  Another feature of this example is that no agent has
a clue for what strategy the other agent is using.  Indeed after $k$ steps, $\Al$ does
not know if $\Be$ has used a strategy derived of equilibria in $\SAcBes$ or in
$\SBcAes$.  In other words, $\Al$ does know if $\Be$ will stop eventually or not and
vice versa. The agents can draw no conclusion of what they observe.  If each agent
does not believe in the threat of the other she is naturally led to escalation.

\subsection*{Acknowledgements}

The author thanks Samson Abramsky, Franck Delaplace, Stephane Leroux, Mat\-thieu
Perrinel, Ren{\'e} Vestergaard, Viktor Winschel for their help, encouragements and
discussions during this research.

\section{Conclusion}
\label{sec:concl}

In this paper, we have shown how to use coinduction in a field, namely economics,
where it has not been used yet, or perhaps in a really hidden form, which has to be
unearthed.  We foresee a future for this tool and a possible way for a refoundation
of economics.


\appendix

\section{Finite 0,1 games and the ``cut and extrapolate'' method}
\label{sec:finite}

We spoke about the ``cut and extrapolate'' method, applied in particular to the
dollar auction.  Let us see how it would work on the 0,1-game.  Finite games, finite
strategy profiles and payoff functions of finite strategy profiles are the inductive
equivalent of infinite games, infinite strategy profiles and infinite payoff functions
which we presented.  Notice that payoff functions of finite strategy profiles are
always defined.  Despite we do not speak of the same types of objects, we use the
same notations, but this does not lead to confusion.  Consider two finite families of
finite games, that could be seen as approximations of the 0,1-game:
\begin{displaymath}
\begin{array}{c@{\quad~~}c}
\begin{array}{lcl}
  F_{0,1} &=& \game{\Al, \game{f_{0,1}}, \game{\Be, \game{f_{1,0}}, F_{0,1}}}  \cup \{\game{f_{0,1}}\}
\end{array}
&%
\begin{array}{lcl}
  K_{0,1} &=& \game{\Al, \game{f_{0,1}}, K_{0,1}'}\\
  K_{0,1}' &=& \game{\Be, \game{f_{1,0}}, K_{0,1}} \cup \{\game{f_{1,0}}\}
\end{array}
\end{array}
\end{displaymath}
In $F_{0,1}$ we cut after $\Be$ and replace the tail by $\game{f_{0,1}}$.  In
$K_{0,1}$ we cut after $\Al$ and replace the tail by $\game{f_{1,0}}$.
Recall~\cite{vestergaard06:IPL} the predicate \textsf{BI}, which is the finite
version of~\textsf{PE}.
\begin{eqnarray*}
  \mathsf{BI}(\game{f})\\
  \mathsf{BI}(\game{p, c, s_d, s_r}) &=& \mathsf{BI}(s_l) \wedge \mathsf{BI}(s_r)
  \wedge \\
  &&\prof{p, d, s_d, s_r} "=>"
\widehat{s_d}(p) \ge \widehat{s_r}(p)\ \wedge\\ && \prof{p, r, s_d, s_r} "=>"
\widehat{s_r}(p) \ge \widehat{s_d}(p)
\end{eqnarray*}
We consider the two families of strategy profiles:

\begin{displaymath}
\begin{array}{c@{\qquad\quad}c}
\begin{array}{rcl}
  \mathsf{SF}_{0,1}(s) &"<=(i)>"& (s = \prof{\Al, d \vee r , \prof{f_{0,1}}, \prof{\Be,
      r, \prof{f_{1,0}}, s'}} \wedge \mathsf{SF}_{0,1}(s')  \quad \vee \\ 
  && s = \prof{f_{0,1}}
  \\\\
  \mathsf{SK}_{0,1}(s) &"<=(i)>"& s = \prof{\Al, r, \prof{f_{0,1}}, s'} \wedge \mathsf{SK}_{0,1}'(s')\\
  \mathsf{SK}_{0,1}'(s) &"<=(i)>"& (s = \prof{\Be, d , \prof{f_{1,0}}, s'} \vee s = \prof{\Be, r , \prof{f_{1,0}}, s'}) \wedge
  \mathsf{SK}_{0,1}(s') \quad \vee \\ &&s = \prof{f_{1,0}}
\end{array}
\end{array}
\end{displaymath}
In $\mathsf{SF}_{0,1}$, $\Be$ continues and $\Al$ does whatever she likes and in
$\mathsf{SK}_{0,1}$, $\Al$ continues and $\Be$ does whatever she likes.  The
following proposition characterizes the backward induction equilibria for games in
$F_{0,1}$ and $K_{0,1}$ respectively and is easily proved by induction:
\pagebreak[2]
\begin{proposition}~
  \begin{itemize}
  \item $\mathsf{game}(s) `:F_{0,1} \wedge \mathsf{SF}_{0,1}(s) "<=>" \mathsf{BI}(s)$,
  \item $\mathsf{game}(s) `: K_{0,1}  \wedge \mathsf{SK}_{0,1}(s) "<=>" \mathsf{BI}(s)$.
  \end{itemize}
\end{proposition}
This shows that cutting at an  even or an odd position does not give the same strategy
profile by extrapolation.  Consequently the ``cut and extrapolate'' method does not
anticipate all the subgame perfect equilibria.  Let us add that when cutting we
decide which leaf to insert, namely $\game{f_{0,1}}$ or $\game{f_{1,0}}$, but we
could do another way around obtaining different results.

\paragraph{0,1 game and limited payroll.}

To avoid escalation in the dollar auctions, some require a \emph{limited payroll},
i.e., a bound on the amount of money handled by the agents, but this is inconsistent
with the fact that the game is infinite.  Said otherwise, to avoid escalation, they
forbid escalation.  One can notice that, in the 0,1-game, a limited payroll would not
prevent escalation, since the payoffs are anyway limited by $1$.
\end{document}